\documentclass[aps,pra,superscriptaddress,twocolumn]{revtex4-1}
\usepackage[utf8]{inputenc}
\usepackage[english]{babel}
\usepackage{amsmath}
\usepackage{amsfonts}
\usepackage{amssymb}
\usepackage{graphicx}
\usepackage{amsthm}
\usepackage{color}
\usepackage{hyperref}
\newcommand{\ket}[1]{\ensuremath{\left|#1\right\rangle}}

\newcommand{\dketbra}[1]{\ensuremath{\left| #1 \right\rangle\left\langle #1 \right|}}

\newcommand{\ave}[1]{\ensuremath{\left\langle#1\right\rangle}}

\newtheorem{theorem}{Theorem}

\begin{document}
\title{The capacity of coherent-state adaptive decoders \\ with interferometry and single-mode detectors}
\author{Matteo Rosati}
\affiliation{NEST, Scuola Normale Superiore and Istituto Nanoscienze-CNR, I-56127 Pisa,
Italy.}
\author{Andrea Mari} 
\affiliation{NEST, Scuola Normale Superiore and Istituto Nanoscienze-CNR, I-56127 Pisa,
Italy.}
\author{Vittorio Giovannetti} 
\affiliation{NEST, Scuola Normale Superiore and Istituto Nanoscienze-CNR, I-56127 Pisa,
Italy.}

\begin{abstract}
A class of Adaptive Decoders (AD's) for coherent-state sequences is studied, including in particular the most common technology for optical-signal processing, e.g., interferometers, coherent displacements and photon-counting detectors. More generally we consider AD's comprising adaptive procedures based on passive multi-mode Gaussian unitaries and arbitrary single-mode destructive measurements. For classical communication on quantum phase-insensitive Gaussian channels with a coherent-state encoding, we show that the AD's  optimal information transmission rate is not greater than that of a single-mode decoder. Our result also implies that the ultimate classical capacity of quantum phase-insensitive Gaussian channels is unlikely to be achieved with the considered class of AD's.
\end{abstract}
\maketitle

\section{Introduction}\label{intro}
Quantum communication theory is a promising field for the application of quantum technology, since its predictions could be applied in the short-term in several settings of practical relevance. An important example is communication on free-space or optical-fiber links, which are well described theoretically by quantum phase-insensitive Gaussian channels~\cite{HolevoBOOK,HolGiovaRev,CAVES}, e.g., the lossy bosonic channel~\cite{EXACT}. \\
The maximum transmission rate of classical information on a quantum channel, known as its capacity, is provided by the Holevo-Schumacher-Westmoreland (HSW) theorem~\cite{schumawest,holevo1,holevo2,holevo3,winter}. In particular for quantum phase-insensitive Gaussian channels the capacity at constrained average input energy can be achieved~\cite{gaussOpt,gaussOpt2,maj1,maj2} by a simple \emph{separable encoding}, i.e., sending sequences of coherent states~\cite{RevGauss}, each of them constituting a letter for a single use of the channel or communication mode. This fact may seem surprising at first, since coherent states are among the simplest states of the electromagnetic field and are often regarded as fundamentally classical. Nevertheless they are sufficient to achieve the maximum communication rate allowed by quantum mechanics on a broad class of channels of considerable practical relevance. Unfortunately the truly quantum challenge posed by these systems seems to reside in the decoding procedures, since all known capacity-achieving measurements require \emph{joint decoding} operations~\cite{hauswoot,schumawest,winter,oga,oganaga,hayanaga,hayashi,seq1,seq2,sen, Arikan,wildeHayden,polarWildeGuha,NOSTRO}, i.e., reading out entire blocks of letters at once  by projecting onto arbitrary entangled superpositions of the codewords. Hence even the classical coherent-state encoding requires a highly non-trivial quantum decoding to achieve capacity. Such joint quantum measurements are difficult to design with current technology~\cite{wildeguha1,wildeguha2,takeokaGuha,takeokaGuha2,lee,Guha1,Guha2,Banaszek,NOSTROHad}, so that the quest for an optimal decoder of separable coherent-state codewords that would finally trigger practical applications is still open. Given the difficulty of implementing truly joint quantum measurements, research has then mainly focused on decoding coherent states with the general class of Adaptive Decoders (AD) depicted in Fig.~\ref{fig1}a. The latter combines the available \emph{single-mode} technology, e.g., photodetectors and local transformations, with \emph{multi-mode passive} interferometers and \emph{classical feedforward control}. The rationale behind this choice is that introducing correlations between modes during the decoding procedure may increase the transmission rate of simple separable measurements, getting closer to the  structure of joint quantum measurements that seems to be ultimately necessary to achieve the capacity of phase-insensitive Gaussian channels. \\
\begin{figure}[t!]
\includegraphics[scale=.24]{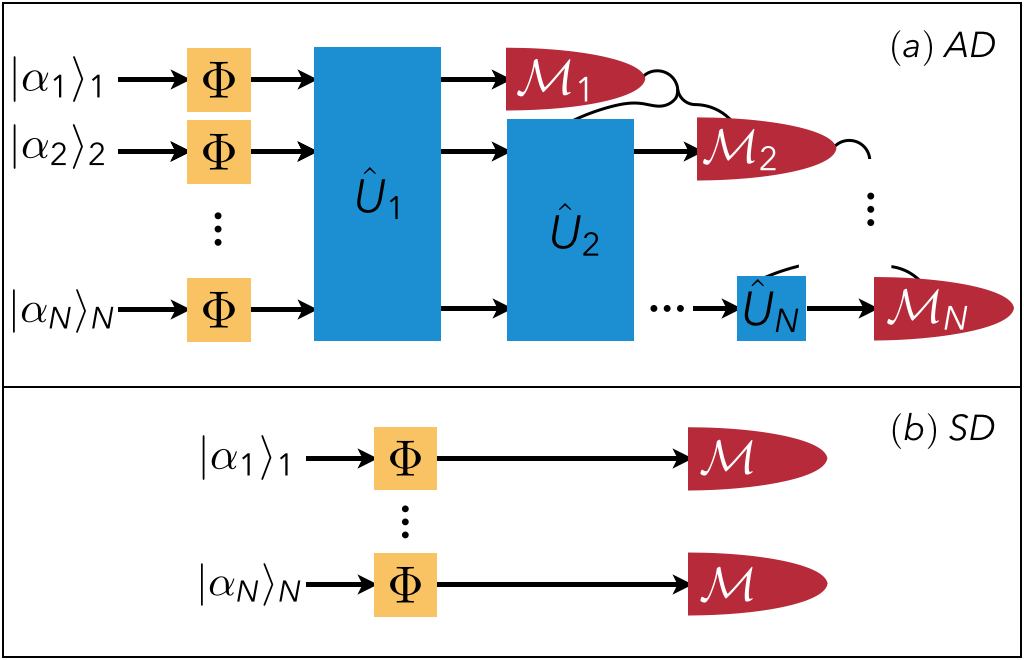}
\caption{Schematic depiction of the class of (a) Adaptive Decoders (AD) and (b) Separable Decoders (SD) considered, whose maximum information transmission rate is proved to be equal. (a) In the AD case the sender, Alice, encodes the message into separable sequences of coherent states $\ket{\alpha_{1}}_{1}\otimes\cdots\otimes\ket{\alpha_{N}}_{N}$ and sends it to the receiver, Bob, with $N$ distinct uses of a quantum phase-insensitive Gaussian channel $\Phi$ (yellow/light-gray boxes), Eq.~\eqref{piGch}. Bob's AD comprises multi-mode passive Gaussian interferometers $\hat{U}_{j}$ (blue/gray boxes), Eq.~\eqref{pass}, and arbitrary destructive single-mode measurements $\mathcal{M}_{j}$ (red/dark-gray shapes), Eq.~\eqref{POVM}, adaptively dependent on the measurement results of previous modes and applied successively on the remaining modes. (b) In the SD case Alice uses the same encoding but Bob performs the same measurement $\mathcal{M}$ on each mode and cannot use adaptive procedures.}\label{fig1}
\end{figure}
On the contrary, in this Letter we prove that the maximum information transmission rate of such channels with coherent-state encoding and AD \emph{is equal to} that obtained with a Separable Decoder (SD) employing the same measurement on each mode, as shown in Fig.~\ref{fig1}b.
The general idea behind our proof is to map the quantum AD into an effective classical programmable channel with feedback to the encoder. Then we obtain our results by extending Shannon's feedback theorem~\cite{ShanFeed,covThom} to this kind of channels.\\
Our work gives several major contributions: i) it implies the conjecture by Chung \emph{et al.}~\cite{guhaDet,guhaDet1}, namely that adaptive passive Gaussian interactions, single-mode displacements and photodetectors do not increase the optimal transmission rate; ii) if the HSW capacity of phase-insensitive Gaussian channels is achieved only by joint measurements, as the evidence suggests so far, then it cannot be achieved with our AD scheme; iii) it extends the results of Takeoka and Guha~\cite{takeokaGuha2}, who considered only Gaussian measurements; iv) it extends the analysis made by Shor~\cite{ShorAd} in the context of trine states to coherent states and passive interactions. Our results, though already envisaged in previous works on the subject, have strong relevance for future research on practical decoders: i) they extend the study of decoders by considering arbitrary single-mode manipulations before measurement, including non-Gaussian and non-unitary ones; ii) they exclude a decoding advantage of adaptive passive Gaussian interactions, which are the easiest to realize in practice, suggesting that more difficult interactions are necessary to achieve capacity. Furthermore the possibility of employing ancillary states is partially included in our AD scheme: this is the case if each ancilla is allowed to interact just with one mode before being measured; otherwise, i.e., if the ancillae can interact with several modes, the problem of determining the decoder's optimal rate remains open and could give a practical advantage over SD's~\cite{nota}. \\
The article is structured as follows: in Sec.~\ref{ad} we describe in detail the communication protocol and the class of decoders considered; in Sec.~\ref{thm} we demonstrate that the AD's optimal rate is equal to the SD's one; in Sec.~\ref{end} we discuss implications and draw our conclusions.

\section{The adaptive decoder}\label{ad}
Let us suppose that the sender, Alice, wants to transmit a classical message on $N$ independent communication modes, employing coherent states of the electromagnetic field. The latter are defined in terms of the field's annihilation and creation operators $\hat{a}$, $\hat{a}^{\dag}$ as displaced vacuum-states of phase-space amplitude $\alpha\in\mathbb{C}$, i.e., $\ket{\alpha}=\hat{D}(\alpha)\ket{0}$, with $\hat{D}(\alpha)=\exp[\alpha\hat{a}^{\dag}-\alpha^{*}\hat{a}]$ the displacement operator. 
The messages, represented by the sequence of classical input random variables $A_{(1,N)}$ with letters $A_{j}=\left\{\alpha_{j}\in\mathbb{C}\right\}$ for each $j=1,\cdots,N$, are encoded into a separable sequence of optical coherent states $\ket{\alpha_{(1,N)}}=\ket{\alpha_{1}}_{1}\otimes\cdots\otimes\ket{\alpha_{N}}_{N}$, one for each mode $\ket{\cdot}_{j}$, where we have used the compact notation $c_{(j,\ell)}$, $j\leq\ell$, to indicate a sequence of quantities $c_{j},\cdots,c_{\ell}$ on different modes, from the $j$-th to the $\ell$-th one. Each message is chosen according to a joint probability distribution $P_{A_{(1,N)}}\left(\alpha_{(1,N)};E\right)$ at constrained average input energy per mode $E$, i.e., 
 \begin{eqnarray}\label{eneCon}\int d^{2N}\alpha_{(1,N)}P_{A_{(1,N)}}\left(\alpha_{(1,N)};E\right)\sum_{j=1}^{N}|\alpha_{j}|^{2}\leq N E.\end{eqnarray}Let us also suppose that the transmission medium is well described by a quantum phase-insensitive Gaussian channel, represented by a linear Completely Positive and Trace Preserving (CPTP) map $\Phi$ on the Hilbert space of a single mode and completely defined by its action on the displacement operator, i.e.,
\begin{eqnarray}\label{piGch}
\hat{D}(\alpha)\xrightarrow{\Phi}\hat{D}(\mu_{1}\alpha)e^{-\mu_{2}\frac{|\alpha|^{2}}{2}},
\end{eqnarray}
in terms of two parameters $\mu_{i}\geq0$ satisfying the constraint $\mu_{2}\geq|1-(\mu_{1})^{2}|$~\cite{HolevoBOOK}. As shown in~\cite{gaussOpt,gaussOpt2,maj1,maj2}, 
the separable coherent-state encoding discussed above achieves the classical capacity of $\Phi$, when its probability distribution is i.i.d. and Gaussian on each mode.\\
The receiver, Bob, has an AD that outputs the sequence of classical random variables $Y_{(1,N)}$, where $Y_{j}=\left\{y_{j}\in\mathcal{I}\right\}$ for all modes $j=1,\cdots,N$ and $\mathcal{I}$ is the set of possible single-mode outcomes, which can be discrete or continuous, e.g., $\mathcal{I}=\mathbb{R}$ for homodyne detection. The probability distribution of the output variables can be computed from the conditional probability of obtaining an outcome sequence $y_{(1,N)}$ if the input sequence $\alpha_{(1,N)}$ was sent, i.e., $P_{Y_{(1,N)}|A_{(1,N)}}\left(y_{(1,N)}|\alpha_{(1,N)}\right)$. The latter is determined by the specific decoding operations of the AD, Fig.~\ref{fig1}a, comprising for all $j=1,\cdots,N$: 
\begin{itemize}
\item a multi-mode passive Gaussian unitary $\hat{U}_{j}\left(y_{(1,j-1)}\right)$, i.e., a network of beam-splitters and phase-shifters conditioned on the outcomes of previous measurements, acting on the set of modes from the $j$-th to the $N$-th as
\begin{eqnarray}\hat{U}_{j}\left(y_{(1,j-1)}\right)\ket{\alpha_{(j,N)}}=\ket{U_{j}\left(y_{(1,j-1)}\right)\alpha_{(j,N)}},\hspace{-.5cm}\label{pass}
\end{eqnarray} 
where $U_{j}$ is the $(N-j+1)$-dimensional unitary matrix representing $\hat{U}_{j}$ in phase-space, applied directly to $\alpha_{(j,N)}$ as a phase-space vector;
\item single-mode operations and a final destructive measurement, altogether represented by a local Positive Operator-Valued Measure (POVM) $\mathcal{M}_{j}\left(\lambda\left(y_{(1,j-1)}\right)\right)$ chosen among a set of possible POVM's that are labeled by the (discrete or continuous) index $\lambda\in\Lambda$ conditioned on the outcomes of previous modes. Each POVM is defined by a collection of positive operators corresponding to the possible single-mode outcomes, 
\begin{eqnarray}\label{POVM}
\mathcal{M}_{j}\left(\lambda\left(y_{(1,j-1)}\right)\right)=\left\{\hat{E}_{y_{j}}\left(\lambda\left(y_{(1,j-1)}\right)\right)\right\}_{y_{j}\in\mathcal{I}},\hspace{-.5cm}
\end{eqnarray} 
where the operators $\hat{E}_{y_{j}}$ sum up to the identity on the Hilbert space of a single mode.
\end{itemize}
For our results to hold, a crucial assumption is that the single-mode POVM's completely destroy the measured state before any information is sent to the rest of the system; if instead Bob can perform partial measurements the AD's rate may increase, see~\cite{ShorAd}. Let us also note that the generic set of allowed POVM's described above can be restricted case by case by properly choosing the $\hat{E}_{y}$. For example the simplest toolbox for optical-signal processing is that of the Kennedy receiver~\cite{Ken} with POVM's of the form
\begin{eqnarray}\label{ken1}
\mathcal{M}^{ken}\left(\lambda\right)&=&\left\{E_{0}(\lambda),\mathbf{1}-E_{0}(\lambda)\right\},\\
E_{0}(\lambda)&=&\hat{D}^{\dag}(\lambda)\dketbra{0}\hat{D}(\lambda),\label{ken2}
\end{eqnarray}
where the index $\lambda\in\mathbb{C}$ is the amplitude of a phase-space displacement in this case. Since the latter depends adaptively on previous outcomes, the AD with a single-mode Kennedy structure behaves similarly to a Dolinar receiver~\cite{Dol}.

\section{The optimal rate} \label{thm}
The performance of a quantum decoder for the transmission of classical information can be evaluated by computing the mutual information of its classical input and output random variables. The latter is defined for our AD as
\begin{eqnarray}\label{mInfo}
I\left(A_{(1,N)}:Y_{(1,N)}\right)=H\left(Y_{(1,N)}\right)-H\left(Y_{(1,N)}\big|A_{(1,N)}\right),\hspace{.5cm}
\end{eqnarray} 
i.e., the difference of the Shannon entropy~\cite{covThom} of $Y_{(1,N)}$ and the Shannon conditional entropy of $Y_{(1,N)}$ given $A_{(1,N)}$. The AD's optimal information transmission rate then is obtained by maximizing the mutual information \eqref{mInfo} over the input distribution with energy constraint $E$ and the decoding operations and regularizing it as a function of the number of uses $N$, i.e., 
\begin{eqnarray}\label{opDef}
\mathcal{R}_{AD}(E)=\lim_{N\rightarrow\infty}~\max_{\substack{ P_{A_{(1,N)}}\left(\alpha_{(1,N)};E\right), \\ \hat{U}_{j}\left(y_{(1,j-1)}\right), \\\lambda\left(y_{(1,j-1)}\right)\in\Lambda} }\frac{I\left(A_{(1,N)}:Y_{(1,N)}\right)}{N}.\hspace{.5cm}
\end{eqnarray}
We want to compare the AD with the SD of Fig.~\ref{fig1}b, comprising for each use of the channel $\Phi$ only a single-mode POVM $\mathcal{M}(\lambda)$ chosen from the same set of those in the AD parametrized by $\lambda\in\Lambda$, Eq.~\eqref{POVM}, but without any interaction or classical communication between modes. Obviously, the optimal rate of this SD is obtained by maximizing the mutual information of the single-mode input and output variables $A_{1}$ and $Y_{1}$ over the input distribution at constrained energy $E$ and the POVM's parameter, i.e.,
\begin{eqnarray}\label{opSep}
\mathcal{R}_{SD}(E)=\max_{P_{A_{1}}\left(\alpha_{1};E\right),~\lambda\in\Lambda}I\left(A_{1}:Y_{1}\right).
\end{eqnarray}
In order to show that the optimization \eqref{opDef} reduces to \eqref{opSep}, we find it useful to consider a more general decoder comprising the AD and a classical feedback link from Bob to Alice, that certainly cannot decrease the optimal rate \eqref{opDef}. Exploiting this feedback and the phase-insensitive property of $\Phi$, Alice can always perform the $\hat{U}_{j}$ instead of Bob. Hence all the AD's interactions are represented by a classical feedback to the encoder, that rearranges the remaining sequences $\alpha_{(j,N)}\in A_{(j,N)}$ into new sequences $\beta_{(j,N)} \in B_{(j,N)}$ with $B_{j}=\left\{\beta_{j}\in\mathbb{C}\right\}$ for all modes $j=1,\cdots,N$, before transmission on the channel. Crucially, each choice of $\hat{U}_{j}$ corresponds to a different rearrangement performed by the encoder in such a way that the total average-energy constraint \eqref{eneCon} is still respected by the joint probability distribution $P_{B_{(1,N)}}(\beta_{(1,N)};E)$ of the new messages $B_{(1,N)}$. \begin{figure}[t!]
\includegraphics[scale=.24]{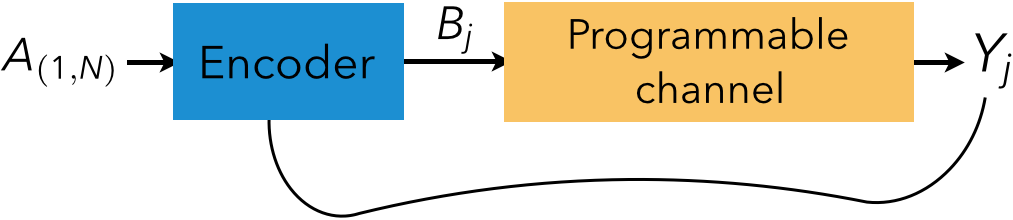}
\caption{Schematic depiction of the classical communication scheme induced by the quantum AD, Fig.~\ref{fig1}a. The input sequence $\alpha_{(1,N)}\in A_{(1,N)}$ is encoded (blue/gray box) into single letters $\beta_{j}\in B_{j}$ that are sent one-by-one on a classical memory channel (yellow/light-gray box) with output $y_{j}\in Y_{j}$, for each $j=1,\cdots,N$. 
The adaptive passive interactions $\hat{U}_{j}$ of the quantum scheme here correspond to a classical feedback to the encoder, i.e., the encoding function that generates each $\beta_{j}$ depends on the input message $\alpha_{(1,N)}$ and on all previous results $y_{(1,j-1)}$. The single-mode phase-insensitive channel $\Phi$ and adaptive POVM $\mathcal{M}_{j}$ employed on each letter $\beta_{j}$ instead correspond to several uses of a classical programmable channel, whose memory at each use depends only on previous outcomes through the parameter $\lambda$ characterizing the measurement, as in Eq.~\eqref{cond}.}\label{fig2}
\end{figure}\\
As a function of the encoded variables $\beta_{j}$, the rest of the AD scheme can be rewritten as a single-mode classical \emph{programmable} channel, i.e., a channel with memory $\lambda$ that can be chosen adaptively depending on previous outcomes. The corresponding conditional probability at the $j$-th use then is
\begin{eqnarray}\label{cond}
P_{Y_{j}|B_{j},Y_{(1,j-1)}}\big(y_{j}\big|&&\beta_{j},\lambda\left(y_{(1,j-1)}\right)\big)\\\nonumber&&=\operatorname{Tr}\left[\hat{E}_{y_{j}}\left(\lambda\left(y_{(1,j-1)}\right)\right)\Phi\left(\dketbra{\beta_{j}}\right)\right],
\end{eqnarray}
where $\hat{E}_{y_{j}}(\lambda)$ are the elements of the POVM $\mathcal{M}_{j}(\lambda)$ as in Eq.~\eqref{POVM}.

In light of the previous observations we can conclude that the AD of Fig.~\ref{fig1}a, with additional classical communication from Bob to Alice, is equivalent to the classical programmable channel \eqref{cond} with feedback, as shown in Fig.~\ref{fig2}. Hence the AD's optimal rate, Eq.~\eqref{opDef}, is upper bounded by the feedback capacity of \eqref{cond}. Similarly, the capacity of the programmable channel without feedback for a single use is equal to the SD's optimal rate, Eq.~\eqref{opSep}.
Eventually, the two classical capacities just defined are related via the following theorem, which is a generalization of Shannon's feedback theorem~\cite{ShanFeed,covThom} to the class of programmable channels considered:
\begin{theorem}\label{feedMemo}
The feedback capacity of a classical programmable channel is equal to its capacity without feedback and it is additive. \end{theorem}
\begin{proof}
Suppose we employ the channel to transmit a classical message $w\in W$ with probability distribution $P_{W}(w)$, outputting $y_{j}\in Y_{j}$ for each use $j$; the most general technique allows a feedback to the sender, who encodes the input message into a sequence of letters $\beta_{j}\in B_{j}$ through an encoding function $\beta_{j}=f\left(w, y_{(1,j-1)}\right)$ for each use $j$. If $\beta$ represents the complex amplitude of a signal we must impose a total average-energy constraint as in Eq.~\eqref{eneCon}. The feedback capacity of this classical programmable channel at constrained total average-energy per mode $E$ is obtained by maximizing the mutual information over the input distribution, the encoding functions and the programmable parameters $\lambda\left(y_{(1,j-1)}\right)$ for each use:  
\begin{eqnarray}\label{multiMode}
C^{fb}_{\infty}(E)=\lim_{N\rightarrow\infty}\max_{\substack{P_{W}\left(w\right),\\ f\left(w, y_{(1,j-1)}\right),\\ \lambda\left(y_{(1,j-1)}\right)\in\Lambda}}\frac{I\left(W:Y_{(1,N)}\right)}{N}.\hspace{.5cm}
\end{eqnarray}
Similarly, for independent uses of the channel without feedback, the capacity at constrained average-energy $E$ can be defined as
\begin{eqnarray}\label{singMode}
C_{1}(E)=\max_{\substack{P_{B_{1}}\left(\beta_{1};E\right),\\ \lambda\in\Lambda}}I\left(B_{1}:Y_{1}\right).
\end{eqnarray}
Now let us note that $C^{fb}_{\infty,\mathcal{C}}(E)\geq C_{1,\mathcal{C}}(E)$, since among all adaptive schemes involved in the optimization \eqref{multiMode} there is one which employs no feedback and the same single-mode measurements that are optimal for Eq.~\eqref{singMode}. To prove the opposite consider the following:
 \begin{eqnarray}
I\big(W:Y_{(1,N)}\big)&=&\sum_{j=1}^{N}I\left(B_{j}:Y_{j}\big|Y_{(1,j-1)}\right)\nonumber\\
&\leq&\ave{\sum_{j=1}^{N}C_{1}\left(E_{j}\left(y_{(1,j-1)}\right)\right)}_{P_{Y_{(1,N)}}\left(y_{(1,N)}\right)}\nonumber\\
&\leq& N C_{1}(E),\label{final}
\end{eqnarray}
where the first equality follows form the chain rule of mutual information and the fact that conditioning over $W$ and $Y_{(1,j-1)}$ is equivalent to conditioning over $B_{j}$ and $Y_{(1,j-1)}$ thanks to the encoding functions, i.e., $H\left(Y_{j}\big|W,Y_{(1,j-1)}\right)=H\left(Y_{j}\big|B_{j},Y_{(1,j-1)}\right)$. The first inequality instead is obtained by employing the definition of Eq.~\eqref{singMode} as an upper bound on each mutual information term in the sum and writing explicitly the average over the output distribution;  the last inequality follows from concavity of the classical capacity as a function of the energy and the total average-energy per mode constraint, i.e., $\sum_{j}E_{j}(y_{(1,j-1)})=NE$.
Eventually by plugging Eq.~\eqref{final} into the definition~\eqref{multiMode} we obtain the upper bound $C^{fb}_{\infty}(E)\leq C_{1}(E)$.
\end{proof}
This implies that the AD's optimal rate is not greater than the SD's one. Since the former is certainly not smaller than the latter, we conclude $\mathcal{R}_{AD}(E)=\mathcal{R}_{SD}(E)$. 

\section{Implications and conclusions}\label{end}
Our analysis implies that a broad class of adaptive decoders for coherent communication on phase-insensitive Gaussian channels, including a majority of those most easily realizable with current technology, cannot beat the optimal single-mode-measurement rate of information transmission. This in turn seems to suggest that such decoders cannot achieve the HSW capacity of phase-insensitive Gaussian channels; however there is no actual proof that joint decoders are really necessary for the task, so that this possibility remains open. In any case our result does not mean that block-coding techniques and adaptive receivers are completely useless for practical applications; indeed in general there may exist specific AD schemes that are more convenient to implement than SD ones and perform equally well, e.g., see Hadamard codes~\cite{Guha1,Guha2,Banaszek,NOSTROHad}.\\
Let us also note that, despite our result is very powerful in decoupling the AD's multi-mode structure for any kind of single-mode POVM, still the difficult optimization of the SD rate of Eq.~\eqref{opSep} is left if one wants an explicit expression of the rate for any set of POVM's. For example we can simplify this calculation for the set of single-mode receivers comprising a coherent displacement followed by any other kind of single-mode operation (the Kennedy receiver of Eq.~\eqref{ken1} belongs to this set). Indeed let us define the variance of a single-mode input probability distribution $P_{A}(\alpha)$ over coherent states as $V=\ave{|\alpha|^{2}}_{P_{A}(\alpha)}-|\ave{\alpha}_{P_{A}(\alpha)}|^{2}$; the energy is instead $E=\ave{|\alpha|^{2}}_{P_{A}(\alpha)}$. One can decide to put a constraint either on the energy or on the variance of the input signals and the former is stricter than the latter. It can then be shown that the net effect of the displacement in a coherent-state receiver is simply to enlarge the family of allowed input distributions from the energy- to the variance-constrained ones so that the optimal rate \eqref{opSep} can be computed on a shrunken set of allowed POVM's. \\
A particularly useful kind of single-mode receivers is that of Kennedy, defined by Eqs.~(\ref{ken1},\ref{ken2}), employing a coherent displacement and an on-off photodetector. The SD's optimal rate for this receiver has been computed in the low-energy limit $E\ll 1$ in~\cite{guhaDet,guhaDet1}, showing that it equals
\begin{equation}
\mathcal{R}^{ken}_{SD}(E)=E\log\frac{1}{E}-E\log\log\frac{1}{E}+O(E).
\end{equation}
Moreover the same authors have shown that an AD scheme without unitaries has the same optimal rate and conjectured that also adaptive unitaries do not help. Our result exactly implies the validity of this conjecture for the particular choice of POVM's~(\ref{ken1},\ref{ken2}). \\
Eventually our result intersects with those of ~\cite{takeokaGuha2,ShorAd}, expanding the set of adaptive receivers whose optimal rate is equal to that of separable ones. Indeed~\cite{takeokaGuha2} compute the capacity of coherent communication with arbitrary adaptive Gaussian measurements, showing it is separable; here instead we considered a restricted interaction set, i.e., passive Gaussians, but an extended single-mode measurement one, i.e., arbitrary POVM's. As for~\cite{ShorAd}, it is stated there that adaptive schemes based on partial single-mode measurements of all the modes may increase the optimal rate; here we considered only destructive single-mode measurements but included the simplest kind of interactions and still could not surpass separable decoding rates.  In particular, as stated in Sec.~\ref{intro}, our AD includes the use of ancillary systems if they interact with just one of the received modes, since it can be thought of as a part of the single-mode destructive measurements. Unfortunately the interaction of ancillary systems with multiple modes is not included, since it results in non-destructive measurements that could provide an advantage over SD's.
Future lines of research could be: studying the lesser-known, interesting class of non-destructive adaptive decoders, computing explicitly the optimal rate for other classes of POVM's, exploring the potential of squeezing and non-Gaussian interactions.


\begin{thebibliography}{99}
\bibitem{CAVES} C. M. Caves and P. D. Drummond, Rev. Mod. Phys. {\bf 66}, 481 (1994).
\bibitem{HolevoBOOK} A. S. Holevo, {\it Quantum Systems, Channels, Information} (de Gruyter Studies in Mathematical Physics, 2012).
\bibitem{HolGiovaRev} A. S. Holevo and V. Giovannetti, Rep. Prog. Phys. \textbf{75}, 046001 (2012).

\bibitem{EXACT}  V. Giovannetti, S. Guha, S. Lloyd, L. Maccone, J. H. Shapiro and H. P. Yuen,   Phys. Rev. Lett. {\bf 92}, 027902 (2004).

  \bibitem{holevo1} A. S. Holevo, Probl. Peredachi Inf. \textbf{9}, 3 (1973); Probl. Inf. Transm. (Engl. Transl.) \textbf{9}, 110 (1973).
  \bibitem{holevo2} A. S. Holevo, IEEE Trans. Inf. Theory \textbf{44}, 269 (1998).
\bibitem{holevo3} A. S. Holevo, e-print arXiv:quant-ph/9809023 [see also Tamagawa University Research Review, no. 4] (1998).
\bibitem{schumawest} B. Schumacher and M. D. Westmoreland, Phys. Rev. A \textbf{56}, 131 (1997); P. Hausladen, R. Jozsa, B. W. Schumacher,
M. Westmoreland, and W. K. Wootters, ibid. \textbf{54}, 1869 (1996).
\bibitem{winter} A. Winter, IEEE Trans. Inf. Theory \textbf{45}, 2481 (1999).

\bibitem{gaussOpt} V. Giovannetti, R. Garc\'{i}a-Patr\'{o}n,	N. J. Cerf	and A. S. Holevo, Nat. Phot. \textbf{8}, 796 (2014).
 \bibitem{gaussOpt2} V. Giovannetti, A. S. Holevo and R. Garc\'{i}a-Patr\'{o}n, Commun.
Math. Phys. \textbf{334}, 1553 (2014).
 \bibitem{maj1} A. Mari, V. Giovannetti and A. S. Holevo, Nature Commun. \textbf{5}, 3826 (2014).
 \bibitem{maj2} V. Giovannetti, A. S. Holevo and A. Mari, Theor. Math. Phys. \textbf{182}, 284 (2015). 

\bibitem{RevGauss} C. Weedbrook, S. Pirandola, R. Garc\'{i}a-Patr\'{o}n, N. J. Cerf, T. C. Ralph, J. H. Shapiro and S. 
Lloyd, Rev. Mod. Phys. \textbf{84}, 621 (2012).


\bibitem{oga}T. Ogawa, Ph.D. dissertation, University of Electro- Communications, Tokyo, Japan, 2000; (in Japanese) T. Ogawa and H. Nagaoka, in 
\textit{Proceedings of the 2002 IEEE International Symposium on Information Theory}, Lausanne, Switzerland, (IEEE, New, York, 2002), p. 73; T. Ogawa, IEEE Trans. Inf.
Theory \textbf{45}, 2486 (1999).
\bibitem{oganaga} T. Ogawa and H. Nagaoka, IEEE Trans. Inf. Theory \textbf{53}, 2261
(2007).
\bibitem{hayanaga} M. Hayashi and H. Nagaoka, IEEE Trans. Inf. Theory \textbf{49}, 1753
(2003).
\bibitem{hayashi} M. Hayashi, Phys. Rev. \textbf{A} 76, 062301 (2007); Commun. Math.
Phys. \textbf{289}, 1087 (2009).


\bibitem{hauswoot} P. Hausladen and W. K. Wooters, J. Mod. Opt. \textbf{41}, 2385 (1994).
\bibitem{seq1} S. Lloyd, V. Giovannetti, L. Maccone, Phys. Rev. Lett. \textbf{106}, 250501 
(2011).
\bibitem{seq2} V. Giovannetti, S. Lloyd and L. Maccone,  Phys. Rev. A 
\textbf{85}, 012302 (2012).
\bibitem{sen} P. Sen, e-print  	arXiv:1109.0802 [quant-ph] (2011).
\bibitem{Arikan} E. Arikan, IEEE Trans. Inf. Theory \textbf{55}, 3051 (2009).
\bibitem{wildeHayden} M. M. Wilde, O. Landon-Cardinal and P. Hayden, in {\it 8th Conference on the Theory of Quantum Computation, Communication and Cryptography (TQC 
2013)}, Dagstuhl, Germany, 2013 arXiv:1302.0398v1.
\bibitem{polarWildeGuha} M. M. Wilde and S. Guha, IEEE Trans. Inf. Theory \textbf{59}, 1175 (2013). 
\bibitem{NOSTRO} M. Rosati and V. Giovannetti, J. Math. Phys. \textbf{57}, 062204 (2016).

 \bibitem{wildeguha1} M. M. Wilde, S. Guha, \textit{Proceedings of the 2012 International Symposium on Information Theory and its Applications}, 303-307 
(2012).
\bibitem{wildeguha2} M. M. Wilde, S. Guha, S.-H. Tan, S. Lloyd, \textit{Proceedings of the 2012 IEEE International Symposium on Information Theory} (ISIT 2012, Cambridge, MA, USA),  
551-555.
\bibitem{takeokaGuha} S. Guha, J. L. Habif and M. Takeoka, \textit{2010 IEEE International Symposium on Information Theory}, 2038-2042 (2010).
\bibitem{takeokaGuha2} M. Takeoka and S. Guha, Phys. Rev. A \textbf{89}, 042309 (2014).
\bibitem{lee} J. Lee, S.-W. Ji, J. Park and H. Nha, Phys. Rev. A \textbf{93}, 050302(R) (2016).
\bibitem{Guha1} S. Guha, Phys. Rev. Lett. \textbf{106}, 240502 (2011).
\bibitem{Guha2} S. Guha, Z. Dutton and J. H. Shapiro, \textit{IEEE International Symposium on Information Theory (ISIT)}, 274	 (2011).
\bibitem{Banaszek} A. Klimek, M. Jachura, W. Wasilewski and K. Banaszek, J. Mod. Opt. \textbf{63}, 2074-2080 (2016). 
\bibitem{NOSTROHad} M. Rosati, A. Mari and V. Giovannetti, Phys. Rev. A \textbf{94}, 062325 (2016).

\bibitem{ShanFeed} C. Shannon, IRE Trans. Inf. Th. \textbf{2}, 8 (1956).
\bibitem{covThom} T. M. Cover and J. A. Thomas, \emph{Elements of Information Theory} (Wiley-Interscience 2006).

\bibitem{guhaDet} H. W. Chung, S. Guha and L. Zheng, \emph{2011 IEEE International Symposium on Information Theory Proceedings} (St. Petersburg, 2011), 284-288.
\bibitem{guhaDet1} H. W. Chung, S. Guha and L. Zheng, e-print arXiv:1610.07578v1 [quant-ph] (2016).
\bibitem{guhaDet1bis} H. W. Chung, S. Guha and L. Zheng, e-print arXiv:1610.07578v3 [quant-ph] (2017).
\bibitem{ShorAd} P. W. Shor, IBM J. Res. Dev. \textbf{48}, 115 (2004).

\bibitem{nota} The latter case of arbitrary interactions with the ancillae has been conjectured as well not to give any advantage in an updated version of the article by Chung et al.~\cite{guhaDet1bis}, after publication of an e-print of the present work~\cite{nostroPreprint}.

\bibitem{Ken} R. Kennedy, MIT Res. Lab. Electron. Quart. Progr. Rep. \textbf{108}, 219 (1973).
\bibitem{Dol} S. Dolinar, MIT Res. Lab. Electron. Quart. Progr. Rep. \textbf{111}, 115 (1973).
\bibitem{nostroPreprint} M. Rosati, A. Mari and V. Giovannetti, e-print arXiv:1703.05701 [quant-ph] (2017).

\end{thebibliography}
\end{document}